\documentclass[sigconf,nonacm]{acmart}

\setcopyright{none}
\copyrightyear{}
\acmYear{}
\acmDOI{}
\acmPrice{}
\acmISBN{}
\acmConference[]{}{}{}

\usepackage[ruled,vlined]{algorithm2e}
\usepackage{graphicx}
\usepackage{textcomp}
\usepackage{xcolor}
\usepackage{listings}
\usepackage{array}
\usepackage{varwidth}
\usepackage{hyperref}
\usepackage{subcaption}
\usepackage{enumitem}
\usepackage{amsthm}
\usepackage{framed}

\usepackage[colorinlistoftodos]{todonotes} 

\usepackage{soul} 

\definecolor{mGreen}{rgb}{0,0.6,0}
\definecolor{mGray}{rgb}{0.5,0.5,0.5}
\definecolor{mPurple}{rgb}{0.58,0,0.82}
\definecolor{backgroundColour}{rgb}{0.95,0.95,0.92}
\definecolor{mygreen}{rgb}{0,0.6,0}
\definecolor{mygray}{rgb}{0.5,0.5,0.5}
\definecolor{mymauve}{rgb}{0.58,0,0.82}

\lstdefinestyle{CStyle}{
  backgroundcolor=\color{backgroundColour}, 
  breakatwhitespace=false,         
  breaklines=true,                 
  captionpos=b,                    
  deletekeywords={...},            
  escapeinside={\%*}{*)},          
  extendedchars=true,              
  keepspaces=true,                 
  morekeywords={*,...},            
  numbers=left,                    
  numbersep=5pt,                   
  numberstyle=\tiny\color{mygray}, 
  rulecolor=\color{black},         
  showspaces=false,                
  showtabs=false,                  
  tabsize=2,	                   
  title=\lstname,                   
  belowcaptionskip=1\baselineskip,
  xleftmargin=\parindent,
  language=C,
  showstringspaces=false,
  basicstyle=\footnotesize\ttfamily,
  keywordstyle=\bfseries\color{green!40!black},
  commentstyle=\itshape\color{purple!40!black},
  identifierstyle=\color{blue},
  stringstyle=\color{orange}
}

\title{An Irredundant Decomposition of Data Flow with Affine Dependences}

\author{Corentin Ferry}
\affiliation{
  \institution{Univ Rennes, CNRS, Inria, IRISA \\ \& Colorado State University}
  \city{Fort Collins, CO}
  \country{USA}}
\email{cferry@mail.colostate.edu}

\author{Steven Derrien}
\affiliation{
  \institution{Univ Rennes, CNRS, Inria, IRISA}
  \city{Rennes}
  \country{France}}
\email{steven.derrien@irisa.fr}

\author{Sanjay Rajopadhye}
\affiliation{
  \institution{Colorado State University}
  \city{Fort Collins, CO}
  \country{USA}}
\email{sanjay.rajopadhye@colostate.edu}

\begin{abstract}
Optimization pipelines targeting polyhedral programs try to maximize the compute
throughput. Traditional approaches favor reuse and temporal locality; while the
communicated volume can be low, failure to optimize spatial locality may cause a
low I/O performance.

Memory allocation schemes using data partitioning such as data tiling can 
improve the spatial locality, but they are domain-specific and rarely applied
by compilers when an existing allocation is supplied. 

In this paper, we propose to derive a partitioned memory allocation for tiled
polyhedral programs using their data flow information. We extend the existing 
MARS partitioning \cite{Ferry_2023} to handle affine dependences, and determine
which dependences can lead to a regular, simple control flow for communications.

While this paper consists in a theoretical study, previous work on data
partitioning in inter-node scenarios has shown performance improvements due to
better bandwidth utilization.
\end{abstract}

\begin{document}

\maketitle

\section{Introduction}


The performance of programs is determined by multiple metrics, among which
execution time and energy consumption. One of the main drivers of these two
metrics is data movement: communication latency causes bottlenecks capping the
compute throughput, and inter-chip communication significantly increases the
power consumption. Optimizing data movement is a tedious task that involves
significant modifications to the program; the extent of programs to optimize
warrants automation. Powerful compiler analyses and abstractions have been
developed in this aim, one of the most powerful of which is the polyhedral
model. 

In the polyhedral model, it is possible to entirely determine the execution
sequence of a program, and its data movement. Optimizations are done in two
respects: first, reducing the amount of communication by exhibiting 
\emph{locality}; second, by optimizing the existing communications to reduce
their latency and better utilize the available bandwidth. 

Techniques improving bandwidth/access utilization \cite{Bondhugula_2013,
Dathathri_2013, Ferry_2023} propose to decompose the data flowing between
statements within a program and group together intermediate results based on
their users; coalescing data accesses allows to better utilize the bandwidth.
However, these data flow optimizations are too restrictive, because they omit
all input data. Like it is done for intermediate results, input data transfers
need to be optimized for both locality and memory access performance.

Furthermore, dependences to input variables are rarely \emph{uniform}, because
the data arrays usually have less dimensions than the domain of computation. The
existing dependence-based partitioning works must therefore be extended to
support affine dependences to input variables, and to propose a re-allocation of
these variables.

This paper seeks to extend the partitioning of \cite{Ferry_2023} to handle the
entire data flow of the tile and maximize access contiguity. Its contributions
are as follows:
\begin{itemize}
  \item We propose a partitioning scheme, called Affine-MARS, of data spaces and 
  iteration spaces with a pre-existing tiling,
  \item We formalize the construction of this partitioning scheme and determine 
  its limitations.
\end{itemize}

This paper is organized as follows: Section~\ref{sec:motivation} justifies this
work on partitioning iteration and data spaces; Section~\ref{sec:background}
gives the notions of MARS and the linear algebra concepts used throughout this
work; Section~\ref{sec:construction-of-mars} proposes construction methods for
Affine-MARS according to the dependences; finally,
Section~\ref{sec:related-work} compares this approach to existing iteration- and
data-space partitioning methods.

\section{Motivation}
\label{sec:motivation}

The motivation of this work stems from two driving forces: the necessity to 
exhibit data access contiguity, and the limitations of existing analyses 
preventing efficient (coalesced) memory accesses. 

\subsection{Necessity of spatial locality}

To motivate this work, we can consider a matrix multiplication program. At each
step of its computations, it needs input values ($a_{i,k}$ and $b_{k,j}$), an
intermediate result (partial sum of $c_{i, j}$) and produces a new result. 
Previous work has shown that using loop tiling increases the performance due
to improved locality. When tiling is applied, the matrices are processed in
``patches'' as illustrated in Figure~\ref{fig:matrix-prod}. 

In this application, multiplying matrices $A=(a_{i,j})$ and $B=(b_{i,j})$ is
done by computing all $a_{i, k} \times b_{k, j}$. Loop tiling, for locality, can
be applied and gives a division of the space as in
Figure~\ref{fig:matrix-prod}.

\begin{figure}
\centering
\includegraphics[width=0.65\columnwidth]{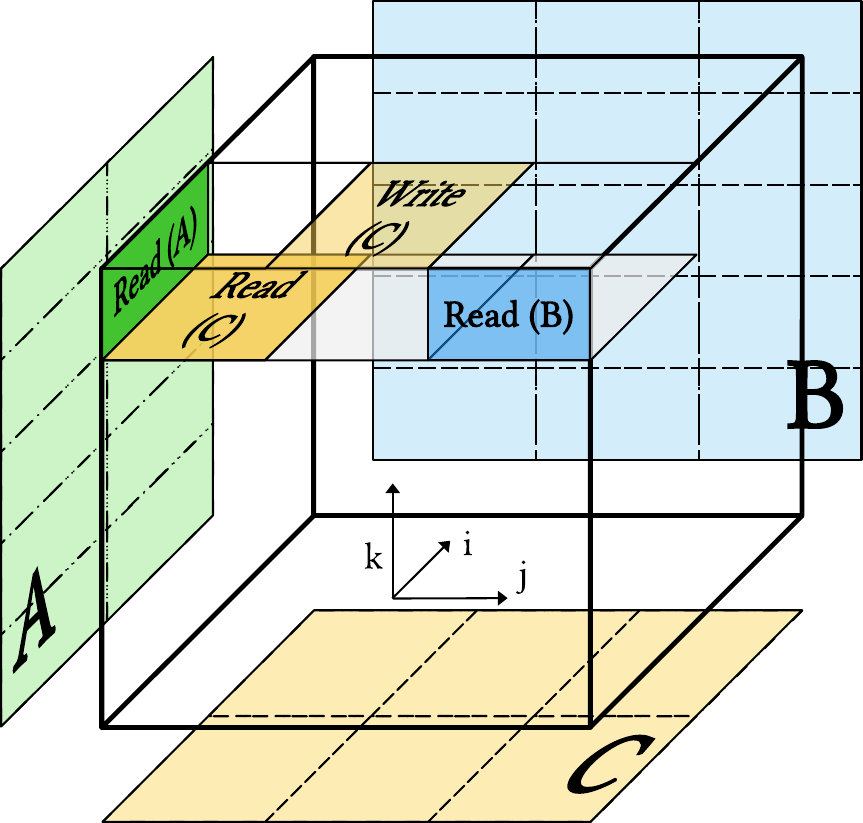}
\caption{Tiled matrix product}
\label{fig:matrix-prod}
\end{figure}

In this example, an entire patch of each input matrix $A$, $B$ is transferred
for the execution of each tile. 

Despite the added locality, the application can still be memory-bound: tiled
matrix product lacks data access contiguity. Barring any data layout
manipulations, data is contiguous within a row (for row-major storage) or column
(for column-major storage). A patch of $A$, $B$ or $C$ is never contiguous
because it contains \emph{multiple parts} of contiguous rows (or columns). The
lack of contiguity therefore induces multiple short burst accesses to retrieve
the entire patch.

Like for intermediate results, it is desirable to increase spatial locality and
leverage contiguity to obtain higher performance on the input variables. Data
blocking has been known to increase the performance of matrix multiplication,
especially because data block correspond exactly to the ``footprint'' of
iteration tiles on the matrix.

\subsection{Limitations of existing transformations}

Although data tiling is sufficient for matrix multiplication, more complex 
computational patterns require a finer data partitioning. 

Ferry et al. \cite{Ferry_2023} proposes a breakup of intermediate results of
programs with purely uniform dependence patterns, that enables contiguity.
However, such dependence patterns exclude commonly found affine dependences,
such as the \emph{broadcast}-type dependences of matrix multiplication, despite
there existing a natural breakup like data blocking. 

Moreover, automatic data blocking is mostly applied by domain-specific compilers
that have to generate the memory allocation (e.g. Halide
\cite{Ragan_Kelley_2013}, AlphaZ \cite{Yuki_2012}). When there exists a memory
allocation and data layout in the input code, compilers follow it unless
specific directives (e.g. the \texttt{ARRAY\_PARTITION} directive in FPGA
high-level synthesis tools based on \cite{Cong_2011_Partitioning}) are given to
them. Allowing the compiler to change this allocation would open the door to
better bandwidth utilization. Works on inter-node communication
\cite{Dathathri_2013,Zhao_2021} where memory allocation only exists within the
nodes (and not across nodes) can resort to very specific groupings of data to
minimize the communication overhead; it makes sense to apply this idea likewise
to host-to-accelerator communications, despite there existing a global memory
allocation.

In this work, we generalize the principle of data blocking to automatically
partition the data arrays in function of when (in time) they are consumed.
This approach can only be guaranteed to work with a specific class of programs
called \emph{polyhedral programs} where the exact data flow is known using
static analysis. 

In the same approach, we propose a secondary partitioning of the intermediate
results; notably, this generalization coincides with an extension of the scope
of previous work \cite{Ferry_2023} to affine dependences.

\section{Background}
\label{sec:background}

We propose an automated approach to partitioning the data flow of a program.
To construct it, we rely the polyhedral analysis and transformation framework
and elements of linear algebra that this section introduces.

\subsection{Polyhedral representation, tiling}
To be eligible for affine MARS partitioning, a program (or a section thereof)
must have a polyhedral representation. It may come either from the analysis of
an imperative program (e.g. using PET \cite{Verdoolaege_2012} or Clan
\cite{Bastoul_2003_Clan}) or a domain-specific language. In any case, the
following elements are assumed to be available:

\begin{itemize}
  \item A $d$-dimensional iteration domain $\mathcal{D} \subset \mathbf{Z}^d$, 
  or a collection of such domains,
  \item A $k$-dimensional data domain $\mathcal{A}$,
  \item A collection $(\varphi_i)_i$ of access functions $\varphi_i : 
  \mathcal{D} \to \mathcal{A}$, defining the  reads and writes at each instance,
  \item A \emph{polyhedral reduced dependence graph} (PRDG), constructed e.g. 
  via array dataflow analysis \cite{Feautrier_1991}.
\end{itemize}

The core elements extracted from the polyhedral representation are the
dependences, that model which data must be available for a computation (any
point in $\mathcal{D}$) to take place. The data flow notably comprises two kinds
of dependences we focus about in this paper:
\begin{itemize}
  \item Flow dependences: correspond to passing of intermediate results within 
  the polyhedral section of the program,
  \item Input dependences: correspond to input data going into the program.
\end{itemize}

Both can be mapped to affine functions corresponding to the following
definition:
\begin{definition}
A \emph{dependence function} is any affine function from an iteration domain
$\mathcal{D}$ to another domain $\mathcal{D}'$ (iteration or data). In
particular, a dependence function is a single-valued relation (each element
of $\mathcal{D}$ has a single image).
\end{definition}
Each dependence will be noted $B$, and as an affine function, it is computed as 
$B(\vec{x}) = A\vec{x} + \vec{b}$ with $A$ a matrix and $\vec{b}$ a vector.

Each \emph{domain} is a subset of a \emph{Euclidean vector space} $E \subset
\mathbf{Z}^d$. In particular, every point $x \in \mathcal{D}$ is associated 
to a vector $\vec{x} \in E$. Section~\ref{sec:lin-alg} gives further elements
of linear algebra used throughout this paper.

\begin{definition}
\label{def:uniform-uniformly-intersecting}
A dependence $B(\vec{x})=A\vec{x}+\vec{b}$ is said to be \emph{uniform} when 
$A$ is the square identity matrix. A collection of dependences $B_1, \dots, B_n$
are \emph{uniformly intersecting} if they all have the same linear part, i.e.
the same $A$ matrix.  
\end{definition}

To create a partitioning of the data spaces, our work relies on an existing
partitioning of the iteration space. Loop tiling \cite{Irigoin_1988,
Schreiber_1990, Wolf_1991, Ramanujam_1992}, a locality optimization, creates
such a partitioning. It uses \emph{tiling hyperplanes} to do so. Each hyperplane
is defined by a \emph{normal vector} (of unit norm). Tiles are periodically
repeated, with a period $s$ called the \emph{tile size}. We notably use
\emph{scaled normal vectors} that translate a point from a tile to the same
point in another tile by crossing one tiling hyperplane.

In this work, we assume \textbf{tiling hyperplanes are linearly independent from
each other}. Each tile has (unique) coordinates that are represented by a
$t$-dimensional vector $\vec{t} = (i_1, \dots, i_t)$ where $t$ is the number of
tiling hyperplanes.
This tile is the set defined by $$ T(\vec{t}) = \left\lbrace \vec{x} \in E :
\bigwedge_{j \in \lbrace 1, \dots, t \rbrace} s_ji_j \leqslant \frac{1}{s_j}
(\vec{x} \cdot \vec{n}_j) < s_j(1 + i_j) \right\rbrace$$

The \emph{footprint} of a dependence $B$ and a tile $T(\vec{t})$ is the image of
the tile by the dependence: $$ B\left<T(\vec{t})\right> = \left\lbrace
B(\vec{x}) : x \in T(\vec{t}) \right\rbrace$$

\subsection{Linear algebra}
\label{sec:lin-alg}

In this paper, we use several fundamental results from linear algebra. Below 
are reminders of them for the reader's reference.

\subsubsection{Spaces and bases}

\begin{definition}
\label{def:orthonormal base}
Let $E$ be an Euclidean vector space of $d$ dimensions with its scalar product
noted $\vec{x}\cdot\vec{y}$. Let $\mathcal{B} = \left(\vec{e}_1, \dots,
\vec{e}_d\right)$ be a basis of $E$. $\mathcal{B}$ is called an
\emph{orthonormal basis} of $E$ when for all $i \neq j$, $\vec{e}_i \cdot
\vec{e}_j = 0$ and for all $i$, $\vec{e}_i \cdot \vec{e}_i = 1$.
\end{definition}

\begin{proposition}
\label{prop:existence-orthonormal-basis}
Any Euclidean space $E$ of $d$ dimensions admits an orthonormal basis.
\end{proposition}
The proof of Proposition~\ref{prop:existence-orthonormal-basis} is done by 
applying the Gram-Schmidt basis orthonormalization to an existing basis.

\begin{definition}
\label{def:vect-supplementary}
The vector space of all linear combinations of a number of vectors $(\vec{e}_1,
\dots, \vec{e}_n)$ is noted $\mathsf{vect}(\vec{e}_1, \dots, \vec{e}_n)$. 
Notably, that space has up to $n$ dimensions, and exactly $n$ dimensions if 
all the $n$ vectors are linearly independent.
\end{definition}

\begin{definition}
\label{def:supplementary}
Two subspaces $S_1$ and $S_2$ of a vector space $E$ are \emph{supplementary}
into $E$ when their intersection is the null vector $\vec{0}$, and there exists
a decomposition of all $\vec{x} \in E$ as $\vec{x}_1 + \vec{x}_2$ with
$\vec{x}_1 \in S_1$ and $\vec{x}_2 \in S_2$. That decomposition is notably
unique.
\end{definition}

\subsubsection{Linear applications}

\begin{definition}
\label{def:linear-application-null-space}

Let $A:E \to F$ be a linear application.
The subspace $K$ of $E$ such that $\forall \vec{x} \in K,
A\vec{x} = \vec{0}$ is called the \emph{null space} of $A$ and is noted
$\mathsf{ker}(A)$.
\end{definition}

\begin{definition}
\label{def:image-preimage}
Let $A:E\to F$ be a linear application. The image of $E$ by $A$ is noted
$A\left<E\right>$. Likewise, the image of a subspace $S \subset E$ by $A$ is
noted $A\left<G\right>$. The preimage of a subspace $T \subset F$ is noted 
$A^{-1}\left<F\right>$.
\end{definition}

\begin{proposition}
\label{prop:supplementary-of-kernel}
If $A:E\to F$ is a linear application, $E$ has $d$ dimensions, and
$\mathsf{ker}(A)$ is its null space, then let $k \leqslant d$ be the
dimensionality of $\mathsf{ker}(A)$.
There exists a $d-k$-dimensional supplementary $I$ of $\mathsf{ker}(A)$ in $E$,
such that:
$$ \forall \vec{x} \in I, (A\vec{x} = \vec{0} \Rightarrow \vec{x} = \vec{0}) $$
\end{proposition}

\section{Partitioning Data and Iteration Spaces}
\label{sec:construction-of-mars}

This section constitutes the core of our work: it proposes a breakup of the
iteration and data spaces based on the same properties as the existing uniform
breakup, detailed in Section~\ref{sec:uniform-deps}. 

The reasoning leading to the MARS starts from a simple, restrictive case (one
single dependence, Section~\ref{sec:one-dep}) and progressively relaxes its
hypotheses (multiple \emph{uniformly intersecting} dependences,
Section~\ref{sec:unif-int-dep} and non-\emph{uniformly intersecting}
dependences, Section~\ref{sec:non-unif-dep}). The last step of the reasoning in
Section~\ref{sec:tiled-spaces} adds the constraint of partitioning an existing
tiled space, which allows to partition intermediate results.

\subsection{Case of uniform dependences}
\label{sec:uniform-deps}

Maximal Atomic irRedundant Sets (MARS) are introduced in \cite{Ferry_2023}.
They are defined as a partition of the \emph{flow-out} iterations of a tile,
such that every element of the partition is the largest set of iterations that
verifies:
\begin{itemize}
  \item Atomicity: consumption of a single element from a MARS implies 
  consumption of the entire MARS.
  \item Maximality: considering all the consumers of a MARS $M_0$ ($C_0$) and 
  all the consumers of another MARS $M_1$ ($C_1$), if $C_0 = C_1$, then 
  $M_0 = M_1$. 
  \item Irredundancy: each element of the MARS space belongs to no more than 
  a single MARS. 
\end{itemize}

While \cite{Ferry_2023} uses the flow-in and flow-out information in the sense
of \cite{Bondhugula_2013}, input data and output data do not belong to this
information. This work instead resorts on the notion of \emph{footprint} from
\cite{Agarwal_1995}; notably, the notion of \emph{flow-in iterations} of a tile
coincides with that of a footprint of a tile (of iterations) on another tile of
iterations.

The properties of MARS constructed with uniform dependences are the same as
those sought in this paper. Merely proposing a partition of the iteration or
data spaces satisfies the irredundancy property; the properties to actually
check from the partitioning are the atomicity and maximality.

\subsection{The problem: uniform versus affine dependences}
\label{sec:comparison-unif-affine}
In the uniform case, MARS can be constructed by enumerating all the
\emph{consumer tiles} of a given tile, i.e. those other tiles that need data
from that tile. The uniformity guarantees that there are a finite number
of consumer tiles, and that all tiles will exhibit the same MARS regardless of
their position in the iteration space (i.e. MARS are invariant by translation
of a tile).

Affine dependences do not guarantee a finite number of consumer tiles; it may 
be parametric or potentially unbounded. Also, it becomes necessary to assert
when the invariance by translation is possible.

In the rest of this section, we will prove, for one and multiple dependences:
\begin{itemize}
  \item The \textbf{existence of a finite set of representatives of all consumer
  tiles}, suitable to determine the MARS partition,
  \item The \textbf{invariance} of the partitioning \textbf{by a translation of 
  a tile}, or conditions to guarantee it.
\end{itemize}

\subsection{Case of a single affine dependence}
\label{sec:one-dep}

The simplest case is when there is a single affine dependence between a tiled
iteration space and a data space. This subsection starts with an example and
explains the general case afterwards.

\subsubsection{Example}
\label{sec:one-dep-example}

To start with, we introduce an example with a single dependence, and
non-canonical tiling hyperplanes. 

\begin{itemize}
  \item Domain: $\left\lbrace (i, j) : 0 \leqslant i < N, 0 \leqslant j < M 
  \right\rbrace$, basis vectors $\vec{e}_i, \vec{e}_j$
  \item Dependence : $S_0(i, j) \mapsto \mathcal{A}(i)$, represented as 
  $B(i, j) = (i)$  (i.e. $B(\vec{x}) = A\vec{x} + \vec{b}$ with 
  $A : (i, j) \mapsto (i)$ and $\vec{b} = \vec{0}$).
  \item Tiling hyperplanes : $H_0 : i + j$, $H_1 : j - i$
  \item Normal vectors: $\vec{n}_1 = (1, 1)$, $\vec{n}_2 = (-1, 1)$; scaled
  normal vectors (w.r.t. tile size): $\vec{\mathbf{n}}_1 = (s/2, s/2)$, 
  $\vec{\mathbf{n}}_2 = (-s/2, s/2)$
  \item Tile size : $s \in \mathbf{N}^{*}$
\end{itemize}


We want to construct the MARS on the $\mathcal{A}$ data space. To do so, we are
going to compute the \emph{footprint} \cite{Agarwal_1995} of a tile onto the
data space along the dependence $B$; then, by noticing that all footprints are a
translation of the same footprint, we will determine parametrically which tiles
have intersecting footprints, and compute the MARS using the same method as
\cite{Ferry_2023}.

We first define a tile of iterations with a parametric set : we call 
$T(i_0, i_1)$ the set : 
$$ T(i_0, i_1) = \left\lbrace (i, j) : si_0 \leqslant i + j < s(1 + i_0) 
\wedge si_1 \leqslant j - i < s(1 + i_1) \right\rbrace$$

The footprint of $T(i_0, i_1)$ by the dependence $B$, appearing in 
Figure~\ref{fig:mars-single-dep-footprint}, is therefore:
$$ B\left<T(i_0, i_1)\right> = \left\lbrace (i) : \exists j : si_0 
\leqslant i + j < s(1 + i_0) \wedge si_1 \leqslant j - i < s(1 + i_1) 
\right\rbrace$$
where the existential quantifier may be removed:
$$ B\left<T(i_0, i_1)\right> = \left\lbrace (i) : s(i_0 - i_1 - 1) < 2i < 
s(i_0 - i_1 + 1) \right\rbrace$$

\begin{figure}
\centering
\includegraphics[width=\columnwidth]{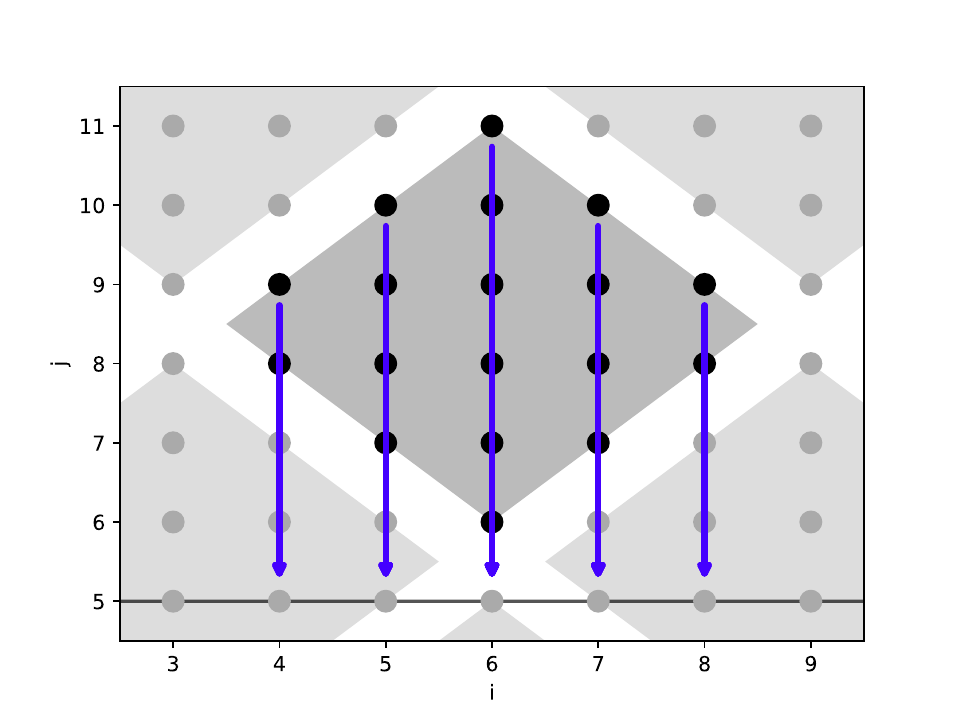}
\caption{Footprint of one tile with a single affine dependence $B(i, j) = (i)$.
The one-dimensional destination space is shown as a continuous line.}
\label{fig:mars-single-dep-footprint}
\end{figure}

Given $(i_0, i_1)$, we now seek the other tiles which footprint's intersection
with $B\left<T(i_0, i_1)\right>$ is not empty: let $(i_2, i_3)$ be another tile.
$$ \begin{aligned}
& B\left<T(i_0, i_1)\right> \cap B\left<T(i_2, i_3)\right> = \\
& \left\lbrace (i) : s(i_0 - i_1 - 1) + 1 \leqslant 2i \leqslant 
s(i_0 - i_1 + 1) - 1 \right. \\
& \left. \wedge s(i_2 - i_3 - 1) + 1 \leqslant 2i \leqslant 
s(i_2 - i_3 + 1) - 1 \right\rbrace
\end{aligned}
$$

The intervals 
$\left[\!\left[s(i_0 - i_1 - 1) + 1; s(i_0 - i_1 + 1) - 1\right]\!\right]$
and \\
$\left[\!\left[s(i_2 - i_3 - 1) + 1; s(i_2 - i_3 + 1) - 1\right]\!\right]$ 
intersect 
%
%
if $i_0 - i_1 = i_2 - i_3 + 1$, $i_0 - i_1 = i_2 - i_3$ or 
$i_0 - i_1 = i_2 - i_3 - 1$.

\begin{figure}
\centering
\includegraphics[width=\columnwidth]{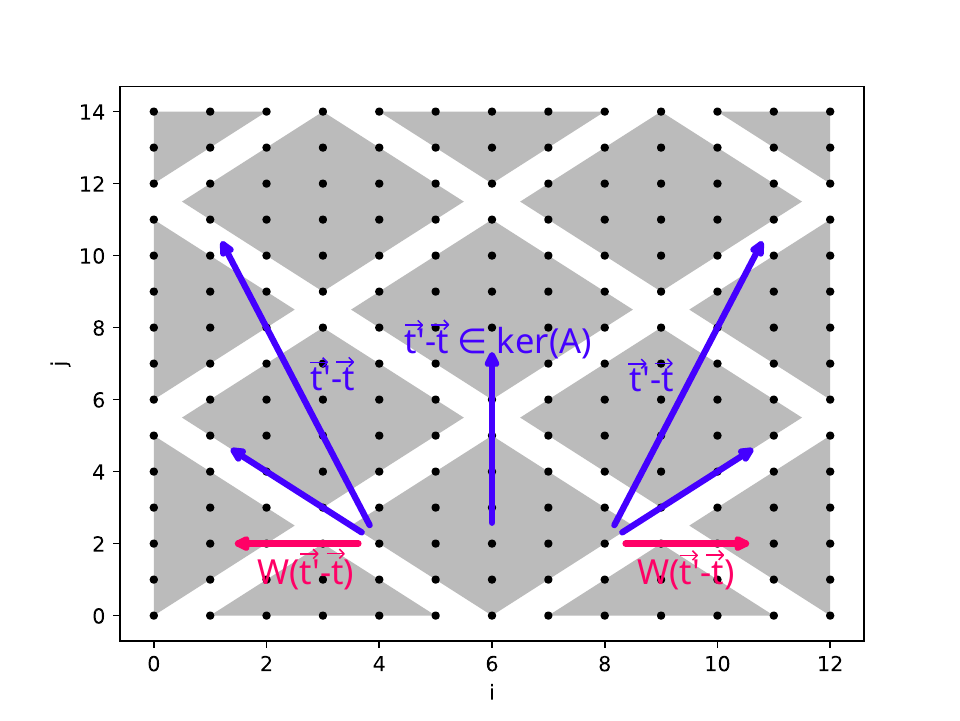}
\caption{Some consumer tiles of one tile $T(\vec{t})$ with a single dependence 
$B(i, j) = (i)$, and the projection $W(\vec{t'}-\vec{t})$ on a supplementary of 
$\mathsf{ker}(A)$. There are a finite number of such projected vectors, and 
they are constant.}
\label{fig:mars-single-dep-consumers}
\end{figure}

The valid $(i_2, i_3)$s are therefore: 
$$ \begin{aligned}
(i_2, i_3) & \in \left\lbrace (i_0 + p, i_1 + p); p \in \mathbf{Z} \right\rbrace \\
& \cup \left\lbrace (i_0 + p - 1, i_1 + p); p \in \mathbf{Z} \right\rbrace \\
& \cup \left\lbrace (i_0 + p + 1, i_1 + p); p \in \mathbf{Z} \right\rbrace
\end{aligned}
$$ as shown in blue in Figure~\ref{fig:mars-single-dep-consumers}.

The space of valid $(i_2, i_3)$ is infinite: we can visually see this as all
tiles along a vertical axis share the same footprint on $\mathcal{A}$.
We can formalize this intuition by computing the kernel of $A$ : in this case,
it is $$\mathsf{ker}(A) = \mathsf{vect}(\vec{e}_j)$$ and the image of a point
on $\mathcal{A}$ is invariant by any upwards or downwards translation.

There are however only three \emph{distinct} footprints intersecting with that
of $T(i_0, i_1)$; the other footprints stem from tiles which are translations
along $\mathsf{ker}(A)$. These footprints come from the top-left, top-right
tiles and all tiles above them vertically; these consumer tiles are shown in 
Figure~\ref{fig:mars-single-dep-consumers}.

We can decompose the space $(i, j)$ using a basis of the kernel and a
supplementary, for instance $E = \mathsf{vect}(\vec{e}_i) \oplus
\mathsf{ker}(A)$.

In this basis, we can express the coordinates of tile's origins for the case 
where $(i_2, i_3) = (i_0 + p + 1, i_1 + p)$ with $p \in \mathbf{Z}$ using the 
scaled normal vectors $\vec{\mathbf{n}}_1$, $\vec{\mathbf{n}}_2$:
$$
\begin{aligned}
i_2\vec{\mathbf{n}}_1 + i_3\vec{\mathbf{n}}_2 &= i_2\frac{s}{2}(\vec{e}_i + \vec{e}_j) + 
i_3\frac{s}{2}(\vec{e}_j - \vec{e}_i) \\
&= \frac{s}{2}(i_0 + p + 1)(\vec{e}_i + \vec{e}_j) + 
\frac{s}{2}(i_1 + p)(\vec{e}_j - \vec{e}_i) \\
&= \frac{s}{2}(i_0 - i_1 + 1)\vec{e}_i + 
\frac{s}{2}(i_0 + i_1 + 2p + 1)\vec{e}_j
\end{aligned}
$$
which, when projected onto $\mathsf{vect}(\vec{e}_i)$, gives:
$$ 
P_{\mathsf{vect}(\vec{e}_i)}(i_2\vec{\mathbf{n}}_1 + i_3\vec{\mathbf{n}}_2) = \frac{s}{2}(i_0 - 
i_1 + 1)\vec{e}_i 
$$
which is independent of $p$. This means that all points within tile 
$T(i_2, i_3)$ have the same image by $B$. Therefore, \textbf{given $(i_0, i_1)$,
the entire family of tiles $(i_0 + p + 1, i_1 + p)$ have the same footprint
on $\mathcal{A}$.} We can therefore consider a single representative of that
family to compute the MARS.

Likewise, if $(i_2, i_3) = (i_0 + p - 1, i_1 + p)$ with $p \in \mathbf{Z}$, then
$$ P_{\mathsf{vect}(\vec{e}_i)}(i_2\vec{\mathbf{n}}_1 + i_3\vec{\mathbf{n}}_2) = \frac{s}{2}(i_0 -
i_1 - 1)\vec{e}_i $$ which is also independent of $p$; the same conclusion holds
for $(i_2, i_3) = (i_0 + p, i_1 + p)$ and
$P_{\mathsf{vect}(\vec{e}_i)}(i_2\vec{\mathbf{n}}_1 + i_3\vec{\mathbf{n}}_2) = \vec{0}$.
Figure~\ref{fig:mars-single-dep-consumers} shows in pink the projection of the
\emph{translation vectors} from the tile $T(\vec{t})$ to its consumers (there
are only two non-null projections, so only two such vectors appear).

There are an infinity of tiles which footprint intersects with that of a given
tile; however, to compute the MARS, we have demonstrated that it is sufficient 
to take three representative tiles. The same procedure as in \cite{Ferry_2023}
can be applied once these three \emph{consumer tiles} have been determined.

Per \cite{Ferry_2023}, we compute the respective intersections with 
$B\left<T(i_0, i_1)\right>$ of all other consumer tiles: for 
$(i_2, i_3) = (i_0 + p + 1, i_1 + p)$ with $p \in
\mathbf{Z}$, $$ \begin{aligned} B\left<T(i_0, i_1)\right> \cap B\left<T(i_0 + p
+ 1, i_1 + p)\right> = \\
\left\lbrace (i) : s(i_0 - i_1 - 1) < 2i < s(i_0 - i_1 + 1) \right. \\
\left. \wedge s(i_0 - i_1) < 2i < s(i_0 - i_1 + 2) \right\rbrace \\
= \left\lbrace (i) : s(i_0 - i_1) < 2i < s(i_0 - i_1 + 1) \right\rbrace \\
\end{aligned}
$$ 

When $(i_2, i_3) = (i_0 + p - 1, i_1 + p)$ with $p \in \mathbf{Z}$, $$
\begin{aligned} B\left<T(i_0, i_1)\right> \cap B\left<T(i_0 + p - 1, i_1 +
p)\right> = \\
= \left\lbrace (i) : s(i_0 - i_1 - 1) < 2i < s(i_0 - i_1) \right\rbrace \\
\end{aligned}
$$ 

Finally, when $(i_2, i_3) = (i_0 + p, i_1 + p)$, $$
\begin{aligned} B\left<T(i_0, i_1)\right> \cap B\left<T(i_0 + p, i_1 +
p)\right> = \\
= \left\lbrace (i) : s(i_0 - i_1) = 2i \right\rbrace \\
\end{aligned}
$$

Also, $B\left<T(i_0 + p + 1, i_1 +p)\right> \cap B\left<T(i_0 + p - 1, i_1
+p)\right> = \varnothing$, so we have all the MARS.

The MARS on symbol $\mathcal{A}$ for this program seen from a tile $T(i_0,
i_1)$ are therefore the three sets $\left\lbrace (i) : s(i_0 - i_1) < 2i < s(i_0 -
i_1 + 1) \right\rbrace$, $\left\lbrace (i) : s(i_0 - i_1 - 1) < 2i < s(i_0 -
i_1) \right\rbrace$ and $\left\lbrace (i) : s(i_0 - i_1) = 2i \right\rbrace$.
These MARS are shown in Figure~\ref{fig:mars-single-dep-mars}.

\begin{figure}
\centering
\includegraphics[width=\columnwidth]{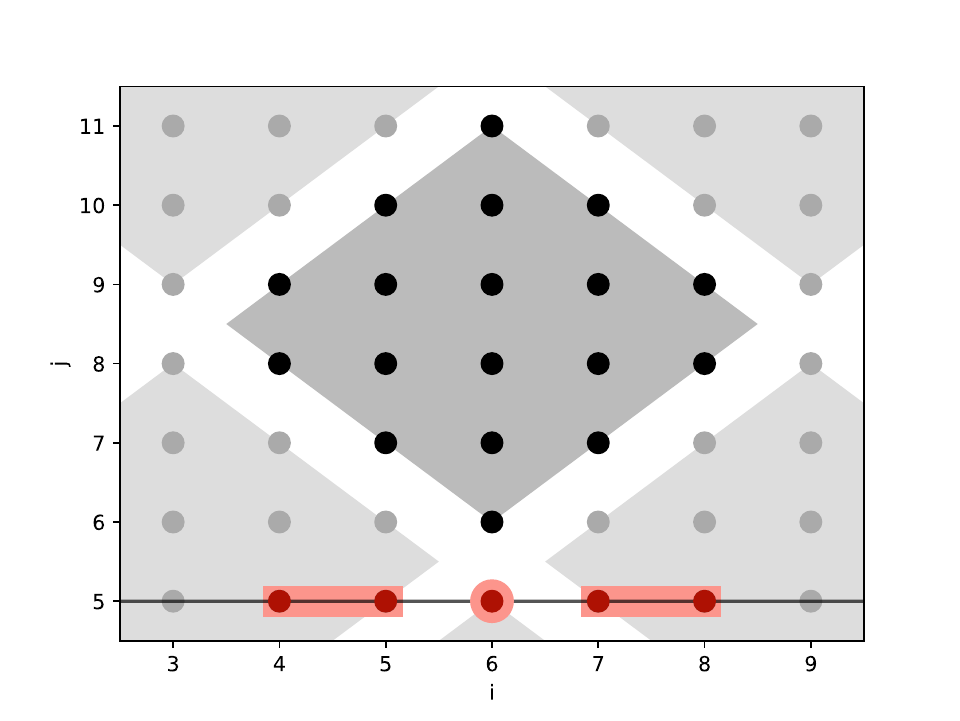}
\caption{MARS obtained with a single affine dependence $B(i, j) = (i)$.}
\label{fig:mars-single-dep-mars}
\end{figure}

\subsubsection{General case}

In the general case, computing the MARS for a single dependence leading to a
non-tiled space can be done as follows.
Let $\mathcal{D}$ be the $N$-dimensional iteration space from which the
dependence originates, and $E$ be the vector space such that $\mathcal{D}
\subset E$.; let $\mathcal{A}$ be the destination space. Let $B$ be the
dependence with $B(\vec{x}) = A\vec{x} + \vec{b}$. Let $(H_1, \dots, H_t)$ be
the $t$ tiling hyperplanes, $(\vec{n}_1, \dots, \vec{n}_t)$ normal vectors to
the tiling hyperplanes, $(s_1, \dots, s_t)$ be the tile sizes.

For a tile coordinate be $\vec{t} = (i_1, \dots, i_t)$, the tile is defined as
$$ T(\vec{t}) = \left\lbrace \vec{x} = (x_1, \dots, x_N) : \forall j \in \lbrace
1, \dots, t \rbrace : s_ji_j \leqslant \vec{x} \cdot \vec{n}_j < s_j(1 + i_j)
\right\rbrace $$

We can compute, with $\vec{t}$ and another tile $\vec{t'}$ as a parameter, when
the intersection of $B\left<T(\vec{t})\right>$ and $B\left<T(\vec{t'})\right>$
is non-empty using affine operations. Let $V(\vec{t})$ be:
$$V(\vec{t}) = \left\lbrace \vec{t'} : B\left<T(\vec{t})\right> \cap
B\left<T(\vec{t'})\right> \neq \varnothing \right\rbrace$$ which is obtainable
by taking the parameter space of $I(\vec{t}, \vec{t'})$. $V(\vec{t})$ represents
the tile coordinates of all tiles which footprint on $\mathcal{A}$ intersects
that of $T(\vec{t})$.

In \cite{Ferry_2023}, $V(\vec{t})$ is determined by browsing through neighboring
tiles. The main difficulty here is that \textbf{$V(\vec{t})$ is potentially
infinite.} We will demonstrate that there are only a finite number of
\emph{distinct footprints} overlapping with $B\left<T(\vec{t'})\right>$. To
determine them, we suggest to decompose $E$ into $\mathsf{ker}(A)$ and a
supplementary $I$ of $\mathsf{ker}(A)$, i.e. $$E = I \oplus \mathsf{ker}(A)$$
Proposition~\ref{prop:supplementary-of-kernel} gives us that such a
decomposition always exists, and per
Proposition~\ref{prop:existence-orthonormal-basis}, there is an orthonormal
basis of the resulting space.

If $r = \mathsf{rank}(A)$, let $(\vec{e}_1, \dots, \vec{e}_r)$ be a basis of $I$
and $(\vec{e}_{r+1}, \dots, \vec{e}_d)$ be a basis of $\mathsf{ker}(A)$ such 
that $(\vec{e}_1, \dots, \vec{e}_d)$ is an orthonormal basis of $E$.
Let $(\vec{n}^p_1, \dots, \vec{n}^p_t)$ be the orthogonal projections of the
$\vec{n}_i$s onto $(\vec{e}_1, \dots, \vec{e}_r)$; in particular, these have
zero $r+1$-th through $d$-th coordinates.

For any $\vec{t'} \in V(\vec{t})$, if $\vec{t'} - \vec{t} = (\delta_1, \dots,
\delta_t)$, we can compute $$ W(\vec{t'} - \vec{t}) = \sum_{i=1}^{t}
\delta_i\vec{n}^{p}_i $$ which represents the part of the translation between
tiles that results in translating the images.

The most important result needed to construct the MARS is the ability to
enumerate \emph{all} the footprints. The following proposition formalizes it:
\begin{framed}
\begin{proposition}
\label{prop:one-dep-cons-tile-finite-uniform}
The set $P(\vec{t}) = \left\lbrace W(\vec{t'} - \vec{t}) :\vec{t'} \in V(\vec{t})
\right\rbrace$ is finite, and for each consumer tile $\vec{t'} \in V(\vec{t})$, 
there exists a unique $\vec{p} \in P(\vec{t})$ such that 
$$B\left<T(\vec{t'})\right> = \left\lbrace \vec{y} + A\vec{p} : \vec{y} \in 
B\left<T(\vec{t})\right> \right\rbrace$$
and that $\vec{p}$ is a constant vector, independent of $\vec{t}$ 
(i.e. the consumer tiles are invariant by translation).
\end{proposition}
\end{framed}
\begin{proof}
\textbf{Completeness of footprints:} Let $\vec{t'} \in V(\vec{t})$, i.e. a tile 
which footprint intersects that of tile $\vec{t}$. We know that 
$W(\vec{t'}-\vec{t}) = \sum_{i=1}^{t}(t'_i - t_i)\vec{n}^p_i = \vec{p} \in 
P(\vec{t})$. Then: 
$$\begin{aligned}
B\left<T(\vec{t'})\right> & = \left\lbrace A\vec{x}+\vec{b} : \forall i : 
s_it'_i \leqslant \vec{x} \cdot \vec{n}_i \leqslant 
(1 + t'_i) s_i\right\rbrace \\
& = \left\lbrace A\vec{x}+\vec{b} : \forall i : s_i(t_i + (t'_i - t_i)) 
\leqslant \vec{x} \cdot \vec{n}_i \leqslant 
(1 + t_i + (t'_i - t_i)) s_i\right\rbrace \\
& = \left\lbrace A\left(\vec{x} + 
{\textstyle \sum_{i=1}^{t}}(t'_i - t_i)\vec{n}_i\right)+\vec{b} : 
\forall i : s_it_i \leqslant \vec{x} \cdot \vec{n}_i \leqslant 
(1 + t_i) s_i\right\rbrace \\
& = \left\lbrace A\left(\vec{x} + 
{\textstyle \sum_{i=1}^{t}}(t'_i - t_i){\color{red}\vec{n}^p_i}\right)+\vec{b} :
 \vec{x} \in T(\vec{t})\right\rbrace \\
& = \left\lbrace \vec{y} + A\vec{p} : \vec{y} \in B\left<T(\vec{t})\right> 
\right\rbrace \\
\end{aligned}
$$
using the fact that ${\color{red} A(\vec{n}_i) = A(\vec{n}_i^p)}$.

\textbf{Uniqueness of $\vec{p}$:} $A$ is bijective between $I$ (supplementary 
of $\mathsf{ker}(A)$ in $E$) and $\textsf{Im}(A)$. Therefore, because
$\vec{p} \in P(\vec{t})$ is in $I$, it is the unique element of $I$ which 
$A\vec{p}$ is the image. Therefore, $\vec{p}$ is unique in the sense of this
proposition.

\textbf{Finiteness of $P(\vec{t})$:} For all $\vec{t'} \in V(\vec{t})$,
$B\left<T(\vec{t'})\right>$ is a translation of $B\left<T(\vec{t})\right>$ by
$A\vec{p}$ with some $\vec{p} \in P(\vec{t})$.
The coordinates of $\vec{t}$ are integers, therefore $A\vec{p}$ is an integer
linear combination of the $A\vec{n}_i^p$s for $i \in \lbrace 1, \dots,
t\rbrace$.
$T(\vec{t})$ being bounded, the footprint $B\left<T(\vec{t})\right>$ is bounded,
therefore only a finite number of translations of itself by $A\vec{p}$s
intersect with it.

\textbf{Constantness of $\vec{p}$:} Let $\vec{t}_0, \vec{t}_1 \in
\mathbf{Z}^{t}$ and $\vec{t'}_0 \in V(\vec{t}_0)$, and
$\vec{p}=W(\vec{t'}-\vec{t})$. Let $\vec{t'}_1 = \vec{t}_1 + (\vec{t'}_0 -
\vec{t}_0)$. Then:
$$\begin{aligned} B\left<T(\vec{t'}_1)\right> &= \left\lbrace A\left(\vec{x} +
{\textstyle \sum_{i=1}^{t}}(t'_{1i} - t_{1i})\vec{n}_i\right)+\vec{b} : \forall
i : s_it_{1i} \leqslant \vec{x} \cdot \vec{n}_i \leqslant (1 + t_{1i})
s_i\right\rbrace \\
 &= \left\lbrace A\left(\vec{x} + {\textstyle \sum_{i=1}^{t}}(t_{1i} - t_{1i} +
 (t'_{0i} - t_{0i}) )\vec{n}_i\right)+\vec{b} : \vec{x} \in T(\vec{t}_1)
 \right\rbrace \\
 &= \left\lbrace A\vec{x} + {\textstyle \sum_{i=1}^{t}}(t'_{0i} -
 t_{0i})A\vec{n}_i + \vec{b} : \vec{x} \in T(\vec{t}_1)\right\rbrace \\
 &= \left\lbrace B(\vec{x}) + A\vec{p} : \vec{x} \in T(\vec{t}_1)\right\rbrace
 \\
\end{aligned}$$ which means that the translation between the images of
$T(\vec{t}_1)$ and $T(\vec{t'}_1)$ is the same as that of $T(\vec{t}_0)$ and
$T(\vec{t'}_0)$.

\end{proof}

We can therefore enumerate $P(\vec{t})$, knowing that for each $\vec{w} \in
P(\vec{t})$, $P^{-1}(\vec{w})$ represents \emph{consumer tiles} that all have
the same footprint by $B$. That footprint is computed as follows: $$
\Phi(\vec{w}) = B\left<T(\vec{t}) + \vec{w}\right> \text{ where } T(\vec{t}) +
\vec{w} = \left\lbrace \vec{x} + \vec{w} : \vec{x} \in T(\vec{t})
\right\rbrace$$

We can then compute the MARS. For all the combinations of $\vec{w}$s, 
i.e. for all $C \in \mathcal{P}\left(P(\vec{t})\right)$, we determine the MARS
associated with that combination of consumer tiles:
$$
\mathcal{M}_C = 
\bigcap_{\vec{w} \in C} \left(\Phi(\vec{w}) \cap B\left<T(\vec{t})\right>\right)
\setminus \bigcup_{\vec{w} \notin C} \left(\Phi(\vec{w}) \cap
B\left<T(\vec{t})\right>\right)$$

\subsection{Case of multiple, uniformly intersecting dependences}
\label{sec:unif-int-dep}

\subsubsection{General case}
If the dependences are uniformly intersecting, they all have the same linear 
part. This means that they all have the same null space, and therefore the
space decomposition into $E = I \oplus \mathsf{ker}(A)$ still applies. 

Let there be $Q$ dependences $B_1, \dots, B_Q$ that are uniformly intersecting.
This means that there exists an unique matrix $A$ such that:
$$ \forall i : B_i(\vec{x}) = A\vec{x} + \vec{b}_i $$

Because all tiles share the same linear part, the space of consumer tiles for
each dependence will be the same \emph{up to a translation}. Their linear part
will notably be the same, and the same argument as in the case above holds to
guarantee that $P(\vec{t}) = \left\lbrace W(\vec{t'} - \vec{t}) :\vec{t'} \in
V(\vec{t}) \right\rbrace$ is finite.

Let, by abuse of the notation, $B\left<T(\vec{t})\right>$ be the combined 
footprint of all dependences:
$$ B\left<T(\vec{t})\right> = \bigcup_{i=1}^{Q} B_i\left<T(\vec{t})\right>$$

For each dependence $B_i$ with $i \in \lbrace 1, \dots, Q\rbrace$, we therefore
compute $ V_i(\vec{t})$ by intersecting $B_i\left<T(\vec{t'})\right>$ and 
$B\left<T(\vec{t})\right>$ (i.e. we want the intersection of the footprint of
\emph{one dependence} and the footprint of \emph{all other dependences}); let $$
V(\vec{t})=\bigcup_{i=1}^{Q}V_i(\vec{t})$$ 
The same decomposition
$E=\mathsf{vect}(\vec{e}_1, \dots, \vec{e}_r) \oplus \mathsf{vect}(\vec{e}_r+1,
\dots, \vec{e}_d)$ is applicable due to all $B_i$s sharing the same linear part
$A$.

We can give a more meaningful expression for $P(\vec{t})$:
$$ P(\vec{t}) = \left\lbrace W(\vec{t'} - \vec{t}) :\vec{t'} \in
\bigcup_{i=1}^{Q} V_i(\vec{t}) \right\rbrace $$
which means that $P(\vec{t})$ is composed of the projections of the vectors 
leading to any consumer tile of \emph{any dependence} (and therefore takes
into account the uniform translations between dependences).

The MARS can be computed by using $P(\vec{t})$. There are two differences with
the case when there is only a single dependence:
\begin{itemize}
  \item The footprints of the consumer tiles $\Phi(\vec{w})$ are specific to 
  each dependence,
  \item The footprint of the tile $\vec{t}$ is the \emph{union} of the footprint
  of all dependences.
\end{itemize}
For $i \in
\lbrace 1, \dots, Q \rbrace$, let $\Phi_i(\vec{w})$ be:
$$\Phi_i(\vec{w}) = B_i\left<T(\vec{t}) + \vec{w}\right> \text{ where }
T(\vec{t}) + \vec{w} = \left\lbrace \vec{x} + \vec{w} : \vec{x} \in T(\vec{t})
\right\rbrace $$

The MARS are constructed by taking all subsets of consumer tiles from
$P(\vec{t})$, and looking at the points consumed \emph{only} by these tiles.

Formally, let the cardinality of $P(\vec{t})$ be $\#C$. For all $K : 1 \leqslant
K \leqslant \#C$ and all permutations $\sigma$ of $\lbrace 1,\dots, \#C
\rbrace$, let $$C= \left\lbrace \vec{t}_{\sigma(1)}, \dots, \vec{t}_{\sigma(K)}
\right\rbrace \text{ and } \overline{C} = \lbrace \vec{t}_{\sigma(K+1)}, \dots,
\vec{t}_{\sigma(\#C)} \rbrace$$ Then, a MARS is constructed according to the
following rules:
\begin{itemize}
  \item For each consumer tile coordinates $\vec{t'} \in C$, there exists a 
  dependence leading to $T(\vec{t'})$,
  \item No dependence leads to a consumer tile $\vec{t'} \in \overline{C}$
\end{itemize}
These two conditions to form a MARS can be written as:
$$\mathcal{M}_{C} = \bigcap_{\vec{w} \in C} \left(\bigcup_{i=1}^{Q}
\Phi_i(\vec{w}) \cap B\left<T(\vec{t})\right>\right) \setminus
\bigcup_{\vec{w} \in \overline{C}} \left(\bigcup_{i=1}^{Q} \Phi_i(\vec{w}) \cap
B\left<T(\vec{t})\right>\right)$$ and there are at most
$\mathsf{card}(\mathcal{P}(P(\vec{t}))) = 2^{\mathsf{card}(P(\vec{t}))}$ $C$s
and therefore as many MARS.

\subsubsection{Example: uniform dependences}
\label{sec:unif-int-dep-example}
In this paragraph, we show that the computation of MARS using \cite{Ferry_2023}
coincides with that proposed in this paper when the dependences are uniform.
Such dependences are a special case of uniformly intersecting dependences,
with a linear part being identity. Note that the destination space is considered
to be a data space, and therefore dependences within a tile are counted in the
footprint (self-consumption of data produced by a tile is dealt with in the
next section). 

Consider the Jacobi 1D example:
\begin{itemize}
  \item Domain: $\left\lbrace (i, j) : 0 \leqslant i < N, 0 \leqslant j < M 
  \right\rbrace$, basis vectors $\vec{e}_i, \vec{e}_j$
  \item Dependences : $B_1(i, j) = (i-1, j-1)$, $B_2(i, j) = (i, j-1)$, 
  $B_3(i, j) = (i+1, j-1)$
  \item Tiling hyperplanes : $H_1 : i + j$ ($\vec{n}_2 = (1, 1)$), $H_2 : j - i$ 
  ($\vec{n}_2 = (-1, 1)$)
  \item Tile size : $s \in \mathbf{N}^{*}$
\end{itemize}

We compute the unified footprint $B\left<T(\vec{t})\right>$:
$$ \begin{aligned} B\left<T(\vec{t})\right> = \left\lbrace (i, j) :
si_1 \leqslant i + j + (2 - p) < s(1 + i_1) \right. &\\
\left. \wedge si_2 \leqslant j - i + p < s(1 + i_2) : p \in \lbrace 0, 1, 2
\rbrace \right\rbrace & \end{aligned}$$ Notably, if we confuse the data space
$\mathcal{A}(i, j)$ and the iteration space $(i, j)$ (that is, each cell of 
$\mathcal{A}$ contains the result of one iteration), and we restrict the
footprint to those points outside tile $T(\vec{t})$, we obtain the
\emph{flow-in} of that tile as in Figure~\ref{fig:mars-uniform-flow-in},
corresponding to the same definition as in \cite{Ferry_2023}.

\begin{figure}
\centering
\includegraphics[width=\columnwidth]{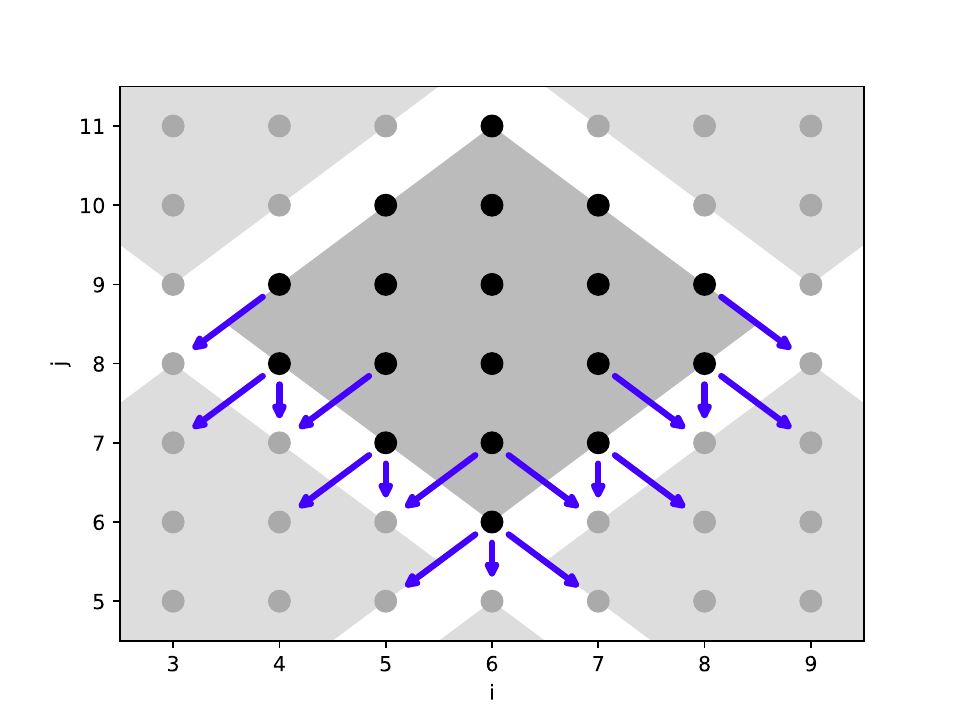}
\caption{Flow-in dependences of tile $T(\vec{t})$ with uniformly intersecting
dependences (Jacobi 1D).}
\label{fig:mars-uniform-flow-in}
\end{figure}

We determine the individual $V_i(\vec{t})$s:
$$
\begin{aligned}
V_1(\vec{t}) &= \left\lbrace (i_1, i_2 - 1), (i_1, i_2), (i_1 + 1, i_2 - 1), 
(i_1 + 1, i_2) \right\rbrace \\
V_2(\vec{t}) &= \left\lbrace (i_1, i_2), (i_1 + 1, i_2), (i_1, i_2 - 1), 
(i_1 - 1, i_2), (i_1, i_2 + 1) \right\rbrace \\
V_3(\vec{t}) &= \left\lbrace (i_1, i_2), (i_1 - 1, i_2), (i_1, i_2 + 1),
(i_1 - 1, i_2 + 1) \right\rbrace \\
\end{aligned}
$$
which gives 
$$ 
\begin{aligned}
V(\vec{t}) = & \left\lbrace (i_1, i_2), (i_1 - 1, i_2), (i_1, i_2 + 1), 
(i_1 - 1, i_2 + 1), (i_1 + 1, i_2), \right.\\
 & \left. (i_1, i_2 - 1), (i_1 + 1, i_2 - 1) \right\rbrace 
\end{aligned}
$$

As $\mathsf{ker}(A) = \lbrace 0 \rbrace$, we easily get that $E =
\mathsf{vect}(\vec{e}_i, \vec{e}_j)$ and therefore constructing the
$W(\vec{t'}-\vec{t})$ is straightforward, yielding the following $P(\vec{t})$:
$$ P(\vec{t}) = \left\lbrace (0, 0), (-1, 0), (0, 1), (-1, 1), (1, 0), (0, -1),
(1, -1) \right\rbrace $$ This $P(\vec{t})$ means there are seven tiles
(including $\vec{t}$ itself) which footprint (i.e. \emph{any dependence})
intersects with $B\left<T(\vec{t})\right>$. These consumer tiles are shown
in Figure~\ref{fig:mars-uniform-intersect-consumers}.
\begin{figure}
\centering
\includegraphics[width=\columnwidth]{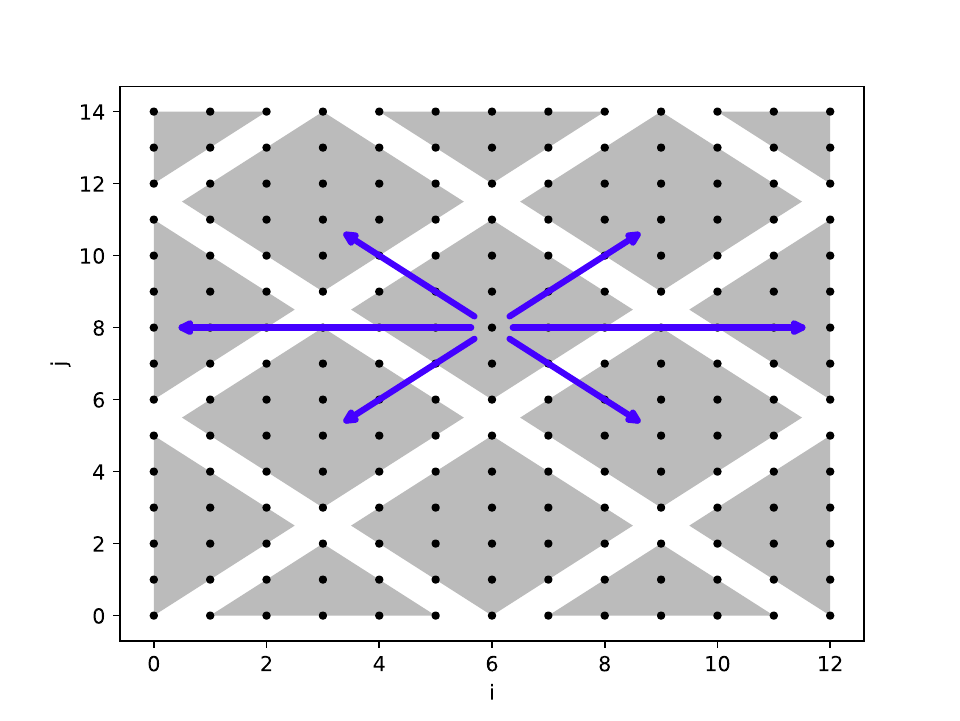}
\caption{Consumer tiles of array $\mathcal{A}$ with Jacobi 1D dependences, 
sharing their footprint with tile $T(\vec{t})$.}
\label{fig:mars-uniform-intersect-consumers}
\end{figure}

For the sake of shortness, we will not enumerate all combinations of consumer
tiles. The MARS that appear after partitioning the footprints stemming from
all consumer tiles are shown in Fig.\ref{fig:mars-uniform-mars-with-tile}.

\begin{figure}
	\centering
	\includegraphics[width=\columnwidth]{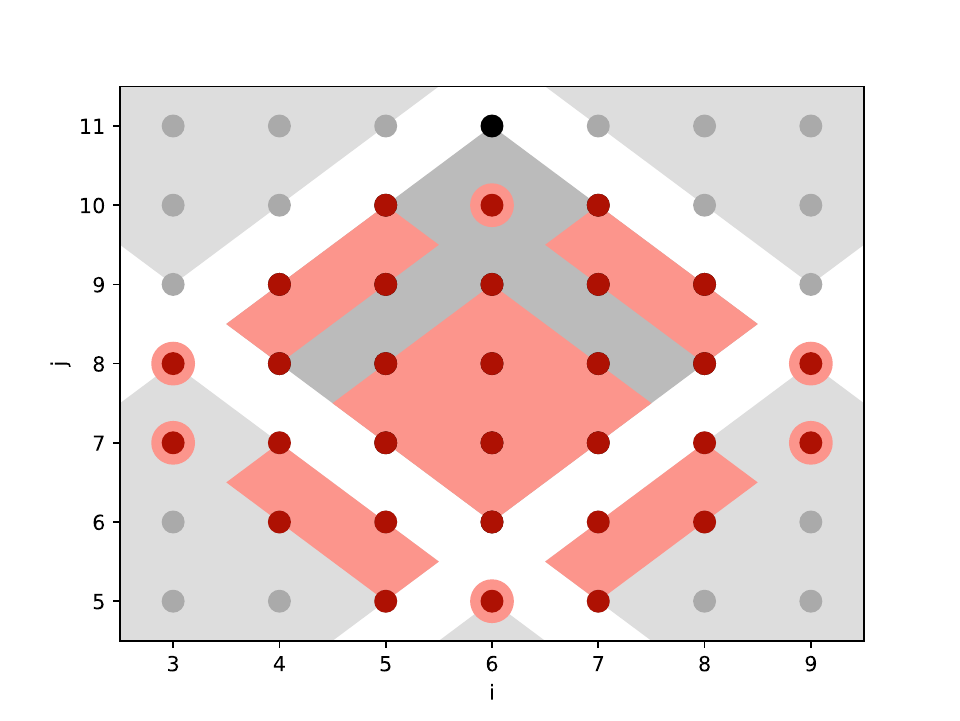}
	\caption{MARS for uniformly intersecting dependences (Jacobi 1D)}
	\label{fig:mars-uniform-mars-with-tile}
\end{figure}

Again, if we confuse the iteration and data spaces, we can obtain the same
MARS as computed in \cite{Ferry_2023} by removing those MARS that are 
contained within $T(\vec{t})$; the result of this operation, shown in 
Figure~\ref{fig:mars-uniform-mars-without-tile}, illustrates the coincidence
between the MARS computed using uniform and affine dependences.
\begin{figure}
	\centering
	\includegraphics[width=\columnwidth]{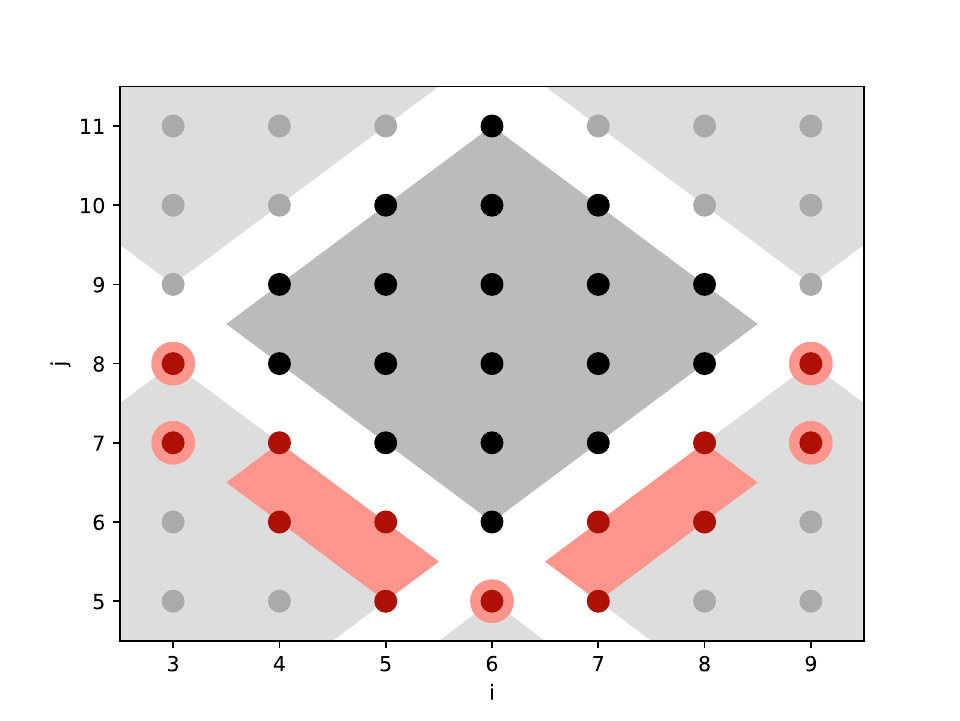}
	\caption{Coincidence between MARS computed with the uniform dependence
	method, and those computed with the affine method.}
	\label{fig:mars-uniform-mars-without-tile}
\end{figure}

\subsection{Case of multiple, non-uniformly intersecting dependences}
\label{sec:non-unif-dep}

We now consider the case where the dependences are not uniformly intersecting.
In this case, the main difference is that dependences no longer share the same
linear part. Therefore, we need to write every dependence separately:
$$ \forall i \in \lbrace 1, \dots, Q \rbrace : B_i(\vec{x}) = A_i\vec{x} +
\vec{b}_i$$
and each dependence having its own null space, there is one orthonormal basis
of the null space and a supplementary per dependence, and therefore one 
projection $W_i(\vec{t'}-\vec{t})$ per dependence.

\subsubsection{Single null space requirement}
\label{sec:non-unif-dep-issue}

Because the dependences may no longer have the same linear part, each linear
part may have a different null space. When considering any consumer tile
$\vec{t'} \in V(\vec{t})$, it is no longer true that the projection of $\vec{t'}
- \vec{t}$ onto each null space is independent of the tile coordinates.
The invariance by translation of a tile from
Proposition~\ref{prop:one-dep-cons-tile-finite-uniform} therefore no longer
holds.

Figure~\ref{fig:mars-non-uniform-shift-issue} gives an example of this case with
two dependences: $B_1(i, j) = (i - j)$ and $B_2(i, j) = (i + j)$. Here,
the $\vec{p}$s depend on $\vec{t}$. Due to the dependence $B_1$, all the tiles
northwest of each tile will intersect with its footprint; but the dependence
$B_2$ generates a footprint southwest, also consumed (because of $B_1$) by all
tiles to its northwest.

In Figure~\ref{fig:mars-non-uniform-shift-issue}, we show the
$W_1(\vec{t'}-\vec{t})$ : $\ker(A_1)$ points to the northeast, and $I_1$
(supplementary of $\ker(A_1)$) points to the northwest, parallel to the 
dependence $B_2$. 

\begin{figure}
	\centering
	\includegraphics[width=\columnwidth]{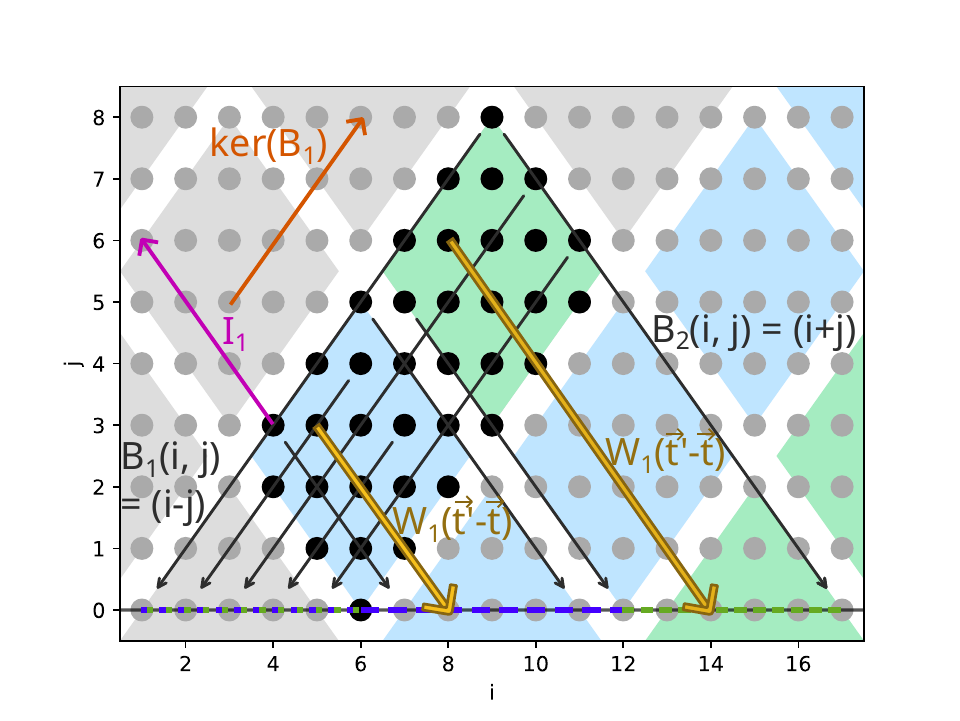}
	\caption{Non-uniformly intersecting dependences do not guarantee that the
	vectors	$W_i(\vec{t'} - \vec{t})$ do not depend on $\vec{t}$, 
	i.e. each tile's footprint is not necessarily a translation of another 
	tile's footprint.}
	\label{fig:mars-non-uniform-shift-issue}
\end{figure}

A sufficient condition for a position-independent footprint to exist is that all
dependences have the same null space:
\begin{proposition}

If all dependences have the same null space, then all tiles have the same
footprint up to a translation. Otherwise said, for any $\vec{\delta} \in
\mathbf{Z}^t$, there exists $\vec{u} \in \mathsf{Im}(B)$ such that:

For each $\vec{t} \in \mathbf{Z}^t$, if $\vec{t'} = \vec{t} + \vec{\delta}$, 
then $$ \bigcup_{i=1}^{Q} B_i\left<T(\vec{t'})\right>
= \left\lbrace \vec{y} + \vec{u} : \vec{y} \in \bigcup_{i=1}^{Q}
B_i\left<T(\vec{t})\right> \right\rbrace$$

\end{proposition}
\begin{proof}
We know that the sought $\vec{u}$ exists for each dependence per
Proposition~\ref{prop:one-dep-cons-tile-finite-uniform}: for each $i\in\lbrace 1, \dots,
Q\rbrace$, there is a $\vec{u}_i$ such that $$B_i\left<T(\vec{t'})\right> =
\left\lbrace \vec{y} + \vec{u}_i : \vec{y} \in B_i\left<T(\vec{t})\right>
\right\rbrace$$

This $\vec{u}_i$ is constructed as:
$$ \vec{u}_i = \sum_{i=1}^{t}(t'_i - t_i)\vec{n}^p_i $$ where the $\vec{n}^p_i$s
are the projections of the normal vectors onto a supplementary of the null space
of each dependence. Because all of the dependences have the same null space, it
comes that all of the $\vec{u}_i$ have the \emph{same} projection onto the
\emph{same} supplementary of that null space. Therefore, they are all equal.
\end{proof}

\subsubsection{Constructon of MARS with a single null space}

We must prove the requirements stated in
Section~\ref{sec:comparison-unif-affine}, proved in the previous two cases,
still hold to compute the MARS.

The uniqueness of $\vec{p}$ (invariance by translation of a tile) has become a
hypothesis, and the dependences must satisfy this requirement to compute MARS.
The previous paragraph only gave a sufficient condition for it to be satisfied.

The finiteness (and enumerability) of the set representatives of consumer tiles
still holds if the dependences all have the same null space.

We can construct the MARS using the same procedure as in \ref{sec:unif-int-dep}:
the footprints of all dependences are distinct, but the null space is the same,
therefore the same definition for $P(\vec{t}) = \left\lbrace W(\vec{t'} -
\vec{t}) \right\rbrace$ as in \ref{sec:unif-int-dep} holds.

\subsubsection{Case of multiple null spaces}

If the dependences have multiple null spaces, there is no guarantee that MARS 
can be constructed (see~\ref{sec:non-unif-dep-issue} for a counter-example).

Due to the null spaces being different, there is no guarantee that:
\begin{itemize}
  \item The consumer tiles which footprint intersects with that of $T(\vec{t})$
  are located at uniform translations of $\vec{t}$ (see counter-example at 
  \ref{sec:non-unif-dep-issue}), and
  \item The projections of translations from $\vec{t}$ to all other consumer 
  tiles onto every null space are finite sets of vectors, i.e. the method used
  previously to obtain a finite set of representative consumer tiles still gives
  a finite set.
\end{itemize}

We propose a solution to the second point: we can obtain finite sets of
translation vectors to represent all consumer tiles for all dependences,
although these translations may be parametric.

The idea is to only consider the projection of dependences that contribute to
a tile's footprint intersecting with $B\left<T(\vec{t})\right>$; and make a
partition of all the sets of consumer tiles according to the contributing
dependences.

Indeed, it is equivalent to say that a dependence $B_i$ contributes
to the consumption, and that the consumer tile $\vec{t'}$ is in $V_i(\vec{t})$.
Multiple dependences contribute to the consumption if and only if $\vec{t'}$ is
simultaneously in all the $V_i(\vec{t})$s of these dependences (i.e. it is in
their intersection).

Using this fact, we can partition the space of all consumer tiles according to
which dependences contribute to each family. In other words, given $D \in
\mathcal{P}(\lbrace 1, \dots, Q\rbrace)$, we compute:
$$ F_D(\vec{t}) = \bigcap_{i \in D} V_i(\vec{t}) \setminus \bigcup_{i \notin D}
V_i(\vec{t})$$
\begin{proposition}
\label{prop:consumer-tile-uniqueness}
Let $\vec{t'} \in \bigcup_{i=1}^{Q} V_i(\vec{t})$. There exists a unique $D \in
\mathcal{P}(\lbrace 1, \dots, Q\rbrace)$ (i.e. a unique $D \subset \lbrace 1,
\dots, Q\rbrace$) such that $\vec{t'} \in F_D$. In other words, $\left\lbrace
F_D : D \subset \lbrace 1, \dots, Q\rbrace \right\rbrace$ is a partition of the
set of all consumer tiles $\bigcup_{i=1}^{Q} V_i(\vec{t})$.
\end{proposition}
\begin{proof}
The construction of the $F_D(\vec{t})$s, given that $\lbrace 1, \dots, Q\rbrace$
is a finite set, guarantees the fact all the $F_D$s are disjoint and therefore
creates a partition of $\bigcup_{i=1}^{Q} V_i(\vec{t})$.
\end{proof}

For each $D \in \mathcal{P}(\lbrace 1, \dots, Q\rbrace)$, $F_D(\vec{t})$
contains a family of consumer tiles (possibly empty). If it is not empty, then
using the same reasoning as in
Proposition~\ref{prop:one-dep-cons-tile-finite-uniform} we can prove the
following proposition:
\begin{proposition}
Let $D \in \mathcal{P}(\lbrace 1, \dots, Q\rbrace)$. Then:
$$P_D(\vec{t}) = \left\lbrace W_i(\vec{t'}-\vec{t}) : \vec{t'} \in F_D \wedge
i\in D\right\rbrace$$ is a finite set.
\end{proposition}
\begin{proof}
Considering that $F_D(\vec{t}) \subset \bigcap_{i \in D} V_i(\vec{t})$,
Proposition~\ref{prop:one-dep-cons-tile-finite-uniform} can be applied to each
individual $V_i(\vec{t})$.
\end{proof}

The effects of parametric vectors leading to consumer tiles are unknown at this
point. Whether MARS can be constructed in this case is left as an open question. 

\subsection{Case of dependences between tiled spaces}
\label{sec:tiled-spaces}

Dependences that lead to tiled spaces correspond to the passing of intermediate
results between tiles. These dependences were supported in \cite{Ferry_2023},
and transmission of intermediate results was done through MARS transiting in
the main memory. This produced a partitioning of the \emph{flow-out set} and
\emph{flow-in set} of each tile. In this section, we extend this principle to
affine dependences.

Uniform dependences used in \cite{Ferry_2023} guaranteed that the producer and 
consumer were in the same space (which is not the case with affine dependences),
and the identity linear part of the dependences gave that the image of a tile by
a dependence was a translation of the tile itself.

The main problem with having different consumer and producer spaces is the
relation between the consumer tiles' ``footprint'' in the producer tiles' space,
and the producer space tiling itself: the footprints of the consumer tiles by
the dependences produce a tiling that may not match with the existing tiling of
the producer space.



In the previous sections (\ref{sec:one-dep}, \ref{sec:unif-int-dep},
\ref{sec:non-unif-dep}), the existence of MARS relied on the footprints of the
consumer tiles (in a tiled iteration space) in the data space (hereafter
\emph{destination space}) being \textbf{independent of the consumer tile} (i.e.
the origin of the dependence). In this section, the destination space is a tiled
iteration space, and we want the tiling induced by the dependence to ``match''
the existing tiling or be finer than it. To this aim, we add the requirement is
that the same footprints are \textbf{independent of the producer tile} (i.e. the
destination of the dependence).

Assuming there are $t$ tiling hyperplanes in the source space, and $q$ tiling
hyperplanes in the destination space, let their (unit) normal vectors be
respectively $\vec{n}_1, \dots, \vec{n}_t$ and $\vec{d}_1, \dots, \vec{d}_q$ and
their tile sizes $s_1, \dots, s_t$ and $z_1, \dots, z_q$.
Let the scaled normal vectors (translation of one tile along each hyperplace) be
$\vec{\mathbf{n}}_1, \dots, \vec{\mathbf{n}}_t$ and $\vec{\mathbf{d}}_1, \dots,
\vec{\mathbf{d}}_q$.

Let there be $Q$ dependences $B_1, \dots, B_q$.
Let $\vec{t}=(t_1, \dots, t_t) \in \mathbf{CT}$ be a \emph{consumer tile} vector
coordinate, and let $\vec{u}=(u_1, \dots, u_q) \in \mathbf{PT}$ be a
\emph{producer tile} vector coordinate (in the destination space of the
dependences).
Let a tile in the destination space be designated as $U(\vec{t})$ using a
definition analogous to that of the source space (see \ref{sec:background}).
Let the consumer tiles of a destination tile $\vec{u} \in \mathbf{DT}$ be
$$X(\vec{u}) = \left\lbrace \vec{t} \in \mathbf{CT} : \exists j \in \lbrace
1, \dots, Q \rbrace : B_j\left<T(\vec{t})\right> \right\rbrace$$ and a tile 
translation vector in the destination space be expressed as:
$$N(\vec{u}) = \sum_{k=1}^{q} u_k \vec{\mathbf{d}}_k$$

The following proposition is a conjecture. It establishes the equivalence
between a translation of a tile in the producer space and the translation of
multiple tiles in the consumer space.

\begin{proposition}
We have the following equivalence:
\begin{gather*}
\forall \vec{u}, \vec{u'} \in \mathbf{DT}, \forall \vec{t} \in X(\vec{u}), \forall i \in \lbrace 1, \dots, Q \rbrace : \\
\exists \vec{t'} \in X(\vec{u'}) : B_i\left<T(\vec{t'})\right> = \left\lbrace \vec{y} + N(\vec{u'} - \vec{u}) : \vec{y} \in B_i\left<T(\vec{t})\right>\right\rbrace \\
\Updownarrow \\
\forall i \in \lbrace 1, \dots, Q \rbrace, \forall j \in \lbrace 1, \dots, t\rbrace, \forall k \in \lbrace 1, \dots, q \rbrace: \\
\exists m \in \mathbf{Z} : m((A_i\vec{\mathbf{n}}_j) \cdot \vec{d}_k)\vec{d}_k = \vec{\mathbf{d}}_k
\end{gather*}
\end{proposition}

If proved, this proposition then establishes a condition on the dependences for
there to be a unique control flow, independent of the tile coordinates, for 
both the MARS to produce (by each producer tile) and the MARS to retrieve (by
each consumer tile).

\section{Related work}
\label{sec:related-work}

This work introduces a partitioning of data arrays and iteration spaces based on
the consumption pattern of each data. Existing work on partitioning aims at
locality in the first place, before memory access optimization. Our work relies
on a locality optimization (tiling) and seeks to further improve memory
accesses. This work is made specific by the combination of its objective 
(partitioning data for spatial locality) and its method (fine-grain 
partitioning where iteration spaces are already partitioned). 

\subsection{Goal of partitioning}
Existing work on partitioning mainly targets locality, such as Agarwal et al.
\cite{Agarwal_1995}. Our work uses the same definitions and follows the same
reasoning as this paper, with a different objective: while \cite{Agarwal_1995}
seeks to adjust the tile size and shapes for locality (i.e. the footprint size
of each tile), we seek to exhibit spatial locality (data contiguity)
opportunities. In that sense, our work is not the first to propose a
partitioning of iteration and data spaces using affine dependences; however, the
desired result (with a spatial locality objective) differs, and so the
construction procedure and hypotheses too.

Parallelism is also an objective: \cite{Zhao_2020} perform iteration and data
space partitioning, then fuse partitions to maximize computation parallelism
while preserving locality. The resulting code is suitable for CPU and GPU
implementation with a cache hierarchy; our partitioning scheme does not follow
the temporal utilization of the retrieved data within a tile. It therefore is
more adapted to scratchpad memories, and scenarios memory accesses can be
decoupled from computations, because grouping data for spatial locality requires
significant on-chip data movement. This makes our approach more suitable for
task-level pipelined (\emph{read, execute, write}) FPGA or ASIC accelerators,
or for small CPU tiles (register tiling) where the register space can be 
considered a scratchpad.

\subsection{Partitioning methods}
Instead of partitioning the inter-tile communicated data with a tiling already
known, one can consider partitioning the inputs and outputs, and deriving tiled
iteration space tiling from the inputs or outputs themselves. This approach is
taken in \cite{Zhao_2020} where the tile shapes are iteratively constructed from
the (tiled) consumers of the iterations or data. 

Monoparametric tiling \cite{Iooss_2021} is performed using an inverse approach
as ours: the data spaces (variables) are partitioned into tiles, and then the
iteration space is partitioned. It requires the program to be represented as a
\emph{system of affine recurrence equations}, where loops do not exist; instead,
iteration spaces start to exist at code generation time, when a variable needs
to be computed. The main difference is that our partitioning scheme must be
applied after loop tiling, and therefore after most locality optimizations.

Dathathri et al. \cite{Dathathri_2013, Bondhugula_2013} partitions the iteration
spaces for inter-node communications in distributed systems, in a manner similar
to MARS: the \emph{flow-out} iterations of each tile are partitioned by
dependences (\emph{dependence polyhedra}) and consumer tiles (\emph{receiving
tiles}). While both approaches are similar with respect to how data is grouped
and transmitted, ours is extended to data space partitioning. Our approach
however adds a restriction on the dependences: we require that the flow-out
partitions are invariant across all tiles, so that a simple, unique control flow
can be derived. Our approach can then be used to create
\emph{position-independent} accelerators that can process any tile in the
iteration space.

It is noteworthy that both our approach and \cite{Dathathri_2013}, along with
other domain-specific inter-node data partitioning schemes (e.g.,
\cite{Zhao_2021}) acknowledge that, to achieve a high bandwidth utilization of
the RAM or network, inter-tile (inter-node) communications need a specific data
layout inferred using static analysis.

\section{Conclusion}

Optimizing programs with respect to memory accesses is a key to improving their
performance. This paper proposes an analysis method to automatically partition
data and iteration spaces from the polyhedral representation of a program when
loop tiling is applied.

Partitioning data arrays is already known to improve spatial locality and, in
turn, access performance. In this paper, we propose a fine-grain partitioning
scheme that can be used to optimize spatial locality.

\bibliographystyle{ACM-Reference-Format}
\bibliography{refs}

\end{document}